\documentclass[10pt,twoside]{amsart}
\usepackage{amsfonts}
\usepackage{amssymb}
\usepackage{graphicx,color}
\usepackage{a4wide}
\usepackage{amsmath}
\usepackage{amsthm}

\date{\today}

\newtheorem{theorem}{Theorem}[section]
\newtheorem{lemma}[theorem]{Lemma}

\newtheorem{Theorem}[theorem]{Theorem}

\theoremstyle{definition}
\newtheorem{definition}[theorem]{Definition}
\newtheorem{example}[theorem]{Example}
\newtheorem{remark}[theorem]{Remark}

\newtheorem{construction}[theorem]{Construction}

\newcommand{\F}{{\mathbb F}}

\newcommand{\Tr}{{\rm Tr}}

\newcommand{\Image}{{\rm Im}}

\newcommand{\be}{\begin{eqnarray}}
\newcommand{\ee}{\end{eqnarray}}
\newcommand{\nn}{{\nonumber}}
\newcommand{\dd}{\displaystyle}
\newcommand{\ra}{\rightarrow}
\newcommand{\Ra}{\Rightarrow}

\begin{document}

\title[Bounds on the covering radius of orthogonal arrays]{\bf New bounds on the covering radius of orthogonal arrays of even strength}

\author[P. Boyvalenkov]{Peter Boyvalenkov} 
\address{Institute of Mathematics and Informatics, Bulgarian Academy of Sciences \\
8 G Bonchev Str., 1113  Sofia, Bulgaria}
\email{peter@math.bas.bg}

\author[F. \"{O}zbudak]{Ferruh \"{O}zbudak}
\address{ Faculty of Engineering and Natural Sciences, Sabancı University, İstanbul, Türkiye }
\email{ferruh.ozbudak@sabanciuniv.edu}

\author[M. Stoyanova]{Maya Stoyanova} 
\address{ Faculty of Mathematics and Informatics, Sofia University ``St. Kliment Ohridski"\\
5 James Bourchier Blvd., 1164 Sofia, Bulgaria}
\email{stoyanova@fmi.uni-sofia.bg}

\begin{abstract} \noindent
We obtain new linear programming (LP) and constructive bounds for the covering radius of binary orthogonal arrays of strength $2k$. Our LP bounds develop in two alternative scenarios. First, if a point $y \in F_2^n$, where the covering radius of some orthogonal array $C \subset F_2^n$ of strength $2k$ is realized, is such that the farthest point of $C$ to $y$ is not antipodal to $y$ we obtain a bound which is better
than the Tiet{\"a}v{\"a}inen (or Fazekas-Levenshtein) bound for non-tight arrays (i.e., the cardinality strictly exceeds the Rao lower bound). Second, if all points where the covering radius is realized are such that their antipodes are in $C$, we obtain a bound which depends on the cardinality of $C$ and is again better whenever the orthogonal array is not tight. We further describe three infinite families of binary orthogonal arrays related to the duals of BCH, Melas, and Zetterberg codes. For these families, we derive lower bounds on the covering radius by applying techniques from algebraic curves over finite fields, while the improved linear programming methods developed in this paper provide upper bounds, leading in some cases to fairly close estimates.
\end{abstract}

\maketitle

\section{Introduction}

Orthogonal arrays have been studied for a wide range of practical applications (experiments, trials, and others) in industry, medicine, agriculture, and others, but also for applications (software
testing, big data, data protection, and others) in computer science (cf. the book \cite{HSS} and references therein). Applications of OAs in cryptography are also considered, for instance, for constructions of secret sharing schemes (see, e.g., \cite{DM93}) and universal hash functions (see, e.g., \cite{CW79,Sti94,GS08}). 

Let $F_2=\{0,1\}$ be an alphabet of two letters and $F_2^n$ be the Hamming space over $F_2$ with the Hamming distance 
$d(x,y)$ between $x,y \in F_2^n$. Let $C \subset F_2^n$ be a code. We again denote by $C$ any $M \times n$ matrix formed by the codewords of $C$ as rows. 

\begin{definition} \label{def-oa} Let $\tau$ and $\lambda$ be positive integers. A code $C \subset F_2^n$ is called a binary orthogonal array (OA) of strength $\tau$ and index $\lambda$ (denoted by OA$_{\lambda}(M,n,\tau)$; $\lambda$ will be omitted since it is implicit from $M$ and $\tau$), if the matrix $C$ has the following property: every $M \times \tau$ submatrix of it contains all ordered $\tau$-tuples
of $F_2^{\tau}$, each exactly $\lambda=M/q^{\tau}$ times as rows.
\end{definition}

There are several equivalent definition of OAs, some of them algebraic. Linear programming techniques use the definition given in Theorem \ref{th:pi-f} below. We also use the important relation $\tau=d^\perp-1$, where $d^\perp$ is the minimum distance of the dual code $C^\perp$ \cite{Del-dis,Del73} (here $F_2$ is the binary field and the duality is Euclidean).

For given $n$, $M$, and $\tau$, we are interested in bounds for the covering radius of binary OA$(M,n,\tau)$. 

\begin{definition} \label{def-cr} 
The covering radius of a code $C \subset F_2^n$ is the quantity
\[ R(C):= \max_{y \in F_2^n} \ \min_{x \in C} d(x,y). \]
In other words, $R(C)$ is the smallest positive integer such that the balls of that radius and centers the codewords of $C$ cover the whole space $F_2^n$.
\end{definition}

The setting in this paper is better suited for a change of the variable $d=d(x,y)$ via $t(d)=1-2d/n$. We will also use the notation $\langle x,y \rangle =1-2d(x,y)/n$. Thus, we consider the quantity
\[ \rho(C):= \min_{y \in F_2^n} \ \max_{x \in C} \left(1-\frac{2d(x,y)}{n}\right)=1-\frac{2R(C)}{n} \]
and will call it covering radius as well.  

Linear programming bounds for covering radius of designs in polynomials metric spaces 
were obtained by Fazekas and Levenshtein \cite[Theorem 2]{FL}. In the binary Hamming spaces $F_2^n$ these bounds coincide with the bounds obtained in 1990-1991 by Tiet{\"a}v{\"a}inen \cite{Tie90,Tie91} and can be formulated as follows. 
If $C \subset F_2^n$ is an OA of strength $\tau=2k-1+e$, $e \in 
\{0,1\}$, just indicates the parity of $\tau$, then
\begin{equation}
\label{FL_bound}
\rho(C) \geq t_k^{0,e},
\end{equation}
where $t_k^{0,e}$ is the largest zero of a Krawtchouk-like polynomial $Q_k^{(0,e)}(t)$ (to be explained below). In terms of distances, \eqref{FL_bound} is written as 
\begin{equation}
\label{Tie_bound}
R(C) \leq d_k^{(n-e)}=\frac{n(1-t_k^{0,e})}{2}
\end{equation}
(this is known as  Tiet{\"a}v{\"a}inen bound), where $d_k^{(n-e)}$ is the smallest zero of the binary Krawtchouk polynomial $K_k^{(n-e)}(z)$ (to be explained below). The values of $t_k^{0,e}$ and $d_k^{(n-e)}$ are shown in a table below (part of Table I on page 282 in \cite{FL}). 
The bounds \eqref{FL_bound} and \eqref{Tie_bound} were investigated in various asymptotic processes \cite{AHLL99,FL,LSS98,LT96,sol95,SS93,Tie90,Tie91}. 

In this paper we use linear programming techniques to obtain a new lower bound for the covering radius of OAs of fixed length $n$ and strength $2k$ which is better than \eqref{FL_bound}. In doing so, we introduce and investigate a class of orthogonal polynomials which are positive definite up to certain degree (that degree is $k-1$ in the setting for strength $2k$). 
Our bounds develop in two alternative scenarios. First, if a point $y \in F_2^n$, where the covering radius of some OA $C \subset F_2^n$ of strength $2k$ is realized, is such that the farthest point of $C$ to $y$ is not antipodal\footnote{Two points in $F_2^n$ are antipodal (each other) if the Hamming distance between them in $n$.} to $y$ we obtain a bound which is better than \eqref{FL_bound}. Second, if all points where the covering radius is realized are such that their antipodes are in $C$, we obtain a bound which depends on the cardinality of $C$ and is again better than \eqref{FL_bound} whenever the orthogonal array is not tight (i.e., its cardinality achieves the Rao lower bound). 

We also describe three constructions of certain infinite families of linear binary orthogonal arrays of even strength. For these OAs, we derive a lower bound (in terms of distances) for their covering radius. A comparison with the linear programming bounds is presented. 

In the first construction, for each positive integer $e$, 
we present an infinite sequence of linear binary orthogonal arrays consisting of $\mathcal{O}^{(1)}(m)$ 
as $m \rightarrow \infty$.
Here, the index $m$ of the sequence corresponds to all positive integers except 
a finite number of small integers depending on $e$. 
In this sequence $\{\mathcal{O}^{(1)}(m)\}_{m \rightarrow \infty}$ of binary orthogonal arrays, 
for the length $n$ and the cardinality $M^{(1)}(m)$ of $\mathcal{O}^{(1)}(m)$ we have
$$
n=2^m-1, \;\; M^{(1)}(m)=2^{me},
$$
and for the strength $\tau$ of $\mathcal{O}^{(1)}(m)$ we have
$$
\tau=2e.
$$
Moreover, we obtain a nontrivial lower bound on the covering radius $R(\mathcal{O}^{(1)}(m))$, which implies that
$$
\limsup_{m \rightarrow \infty} \frac{\log_2 R(\mathcal{O}^{(1)}(m))}{\log_2 M^{(1)}(m)} \ge \frac{1}{e}.
$$

In the second construction, for each integer $m \ge 4$, 
we present a linear binary orthogonal array $\mathcal{O}^{(2)}(m)$.
For the length $n$ and the cardinality $M^{(2)}(m)$ of $\mathcal{O}^{(2)}(m)$ we have
$$
n=2^m-1, \;\; M^{(2)}(m)=2^{2m},
$$
and for the strength $\tau$ of $\mathcal{O}^{(2)}(m)$ we have
$$
\tau=\left\{
\begin{array}{ll}
2 & \mbox{if $m$ is even}, \\
4 & \mbox{if $m$ is odd}.
\end{array}
\right.
$$

In the third construction, for each integer $m \ge 2$, 
we present a linear binary orthogonal array $\mathcal{O}^{(3)}(m)$.
For the length $n$ and the cardinality $M^{(3)}(m)$ of $\mathcal{O}^{(3)}(m)$ we have
$$
n=2^{2m}+1, \;\; M^{(3)}(m)=2^{4m},
$$
and for the strength $\tau$ of $\mathcal{O}^{(3)}(m)$ we have
$\tau=4$.

Again, we obtain a nontrivial lower bounds on the covering radius $R(\mathcal{O}^{(2)}(m))$ and $R(\mathcal{O}^{(3)}(m))$, which imply that
$$
\limsup_{m \rightarrow \infty} \frac{\log_2 R(\mathcal{O}^{(j)}(m))}{\log_2 M^{(j)}(m)} \ge \frac{1}{2}
$$
for $j=2,3$.

These three constructions give rise to three infinite families of binary orthogonal arrays associated with the duals of BCH, Melas, and Zetterberg codes. Determining the exact covering radii of these three families appears to be a natural and difficult problem, and remains open in general. To study this question, we use two different approaches. For the lower bounds, we employ techniques based on algebraic curves over finite fields, together with more detailed methods adapted to the particular families considered here. For the upper bounds, we use the improved linear programming techniques for the covering radius developed below in this paper. In our view, both methods are of independent interest. Combined together, they yield both lower and upper bounds for the covering radii of these families, and in some cases the resulting bounds are quite close.

The paper is organized as follows. In Section 2 we prepare for linear programming (LP) as we explain Krawtchouk and adjacent polynomials, their relations with the OAs, the LP bound for covering radius of OAs before turning to the more specific positive definite signed measures and related orthogonal polynomials. Lemma \ref{lem_pos_def} establishes that the relevant signed measure $\mu_n^{(\ell)}(t)$ is positive definite up to degree $k-1$, where $k$ will be defined later as the half of the OA's strength $\tau$. Section 3 is devoted to the properties of the orthogonal polynomials related to the measure $\mu_n^{(\ell)}(t)$. In Section 4 we obtain our LP bound in two cases. The combination of the two cases presents a universal (in the sense of Levenshtein) bound. In Section 6 we provide three constructions of binary OAs (dual BCH codes, Melas codes, and Zetterberg codes) and derive lower bounds (in terms of distances) for the covering radii of their duals. We compare the LP bounds with the bounds of our constructions for small lengths. 

\section{Preliminaries}

\subsection{Krawtchouk polynomials}

Let $n \geq 2$ be a positive integer. The (binary) Krawtchouk polynomials are defined as
\[ K_i^{(n)}(z):=\sum_{j=0}^i (-1)^j {z \choose j} {n-z \choose i-j}, \  i=0,1,\ldots,n, \]
where $\binom{z}{j} := z(z-1)\cdots(z-j+1)/j!$, $z \in \mathbb{R}$.

The polynomials $K_i^{(n)}(z)$ satisfy the following three-term recurrence relation
\[ (i+1)K_{i+1}^{(n)}(z) = (n-2z)K_i^{(n)}(z) - 
(n-i+1)K_{i-1}^{(n)}(z),	\]
where $K_0^{(n)}(z)=1$ and $K_1^{(n)}(z)=n-2z$. 

We consider the variable change $z=n(1-t)/2$ (inverse to $t=1-2z/n$ used above) to map the distances in $F_2^n$ (the set 
$\{0,1,\ldots,n\}$) in the interval $[-1,1]$; i.e., to the set 
\[ T_n:=\left\{t_i:=-1+\frac{2i}{n}: i=0,1,\ldots,n\right\} \subset [-1,1]. \]
Denoting 
\[ \langle x,y \rangle = t(d(x,y)) :=1-\frac{2d(x,y))}{n}, \ x,y \in F_2^n. \]
we see that $d(x, y) = i$ is equivalent to $\langle x,y \rangle = t_{n-i} \in T_n$. In what follows, we will work with the numbers $\langle x,y \rangle \in T_n$ instead of the Hamming distances $d(x,y)$. 

For the necessary linear programming technique in the interval $[-1,1]$ we switch to the polynomials 
\begin{equation}\label{Kraw} 
Q_i^{(n)}(t) :=\frac{1}{r_i} K_i^{(n)}(z), \nn
\end{equation}
where $r_i:= \binom{n}{i}$, $i=0,1,\ldots,n$.  
In what follows, we will omit the index $(n)$ in the notation and will call $Q_i(t)$ again Krawtchouk polynomials. 

The polynomials $\{Q_i(t)\}_{i=0}^n$ form a basis of the space  of real polynomials of degree at most $n$ and satisfy the three-term recurrence relation 
\[ nt Q_i(t)=(n-i)Q_{i+1}(t)+i Q_{i-1}(t), \ i=1,2,\ldots,n-1, \] 
with the initial conditions $Q_0(t)=1$ and $Q_1(t)=t$. 

The discrete measure of orthogonality for $\{ Q_i(t) \}_{i=0}^n$ is given by 
\begin{equation} \label{KrawOrtho} 
\mu_n  := \frac{1}{2^n} \sum_{i=0}^n r_{i} \delta_{t_i}, \nn
\end{equation}
where $\delta_{t_i} $ is the Dirac-delta measure at the point $t_i \in T_n$. The form 
%\begin{equation}\label{InnerProd}
\[ \langle f,g \rangle=\int_{-1}^1 f(t) g(t) d\mu_n (t)= \frac{1}{2^n} \sum_{i=0}^n r_i f(t_i)g(t_i) \]
%\end{equation}
defines an inner product over the class of polynomials of degree at most $n$ as polynomials of degree at least $n+1$ are reduced modulo $\prod_{i=0}^n(t-t_i)$. 

With this inner product, every polynomial $f(t)$ of degree $r \leq n$ can be uniquely expressed as 
\begin{equation}\label{kr-exp}
f(t)= f_0+\sum_{j=1}^r f_jQ_j(t), 
\end{equation}
where the coefficients $f_i$, $i=0,1,\ldots,r$, can be computed by the formulas
\[ f_j=\frac{1}{2^n} \sum_{i=0}^n {n \choose i} f(t_{n-i}) Q_i(t_{n-j}). \]
In particular, for the linear programming setting, the most important coefficient $f_0$ can be written as
\[ f_0=\frac{1}{2^n}\sum_{i=0}^n {n \choose i} f(t_{n-i}). \]

\subsection{Orthogonal arrays and Krawtchouk polynomials}

The Krawtchouk polynomials are related to the structure of orthogonal arrays as they provide some rules on the distance distributions of OAs.  

\begin{definition}
Let $C \subset F_2^n$ and $y \in F_2^n$. The distance distribution of $C$ with respect to $y$ is the $(n+1)$-tuple
\[ w = w(y) = [w_0(y),w_1(y), \ldots, w_n(y)], \]
where $w_i(y) = |\{x \in C : \ \langle x,y \rangle =t_{n-i} \iff  d(x, y) = i\}|$, $i = 0, \ldots, n$.
\end{definition}

The distance distribution concept allows us to utilize a linear programming approach to the covering radius of orthogonal arrays (cf., e.g., \cite{Tie90,Tie91,FL}). The next theorem of Delsarte gives the necessary algebraic characterization of the orthogonal arrays. 

\begin{theorem} \cite{Del-dis}
\label{th:pi-f} 	
Let $C$ be a binary $OA(M,n,\tau)$ and $y \in F_2^n$. If $w(y)=(w_0,w_1,\ldots,w_n)$ is the distance distribution of $C$ 
with respect to $y$, then for any polynomial $f(t)=f_0+\sum_{j=1}^{\tau} f_jQ_j(t)$ of degree at most $\tau$, the following hold
	\begin{equation} \label{dd-sys}
	\sum_{i=0}^n w_i f(t_{n-i})=f_0 M.
	\end{equation}
The converse is also true; i.e., if \eqref{dd-sys} holds for every $y \in F_2^n$ and every $f$ of degree at most $\tau$, then $C$ is a binary $OA(M,n,\tau)$.
\end{theorem}

The relation \eqref{dd-sys} provides a system of linear equations related to the distance distribution of $C$ with respect to any point $x$ (i.e., to the structure of $C$ with respect to $x$). With a suitably chosen $x$ one can look at the covering radius of $C$. This idea was first developed by  Tiet{\"a}v{\"a}inen \cite{Tie90,Tie91} and generalized for polynomials metric spaces by Fazekas-Levemshtein \cite{FL}. 

\subsection{Linear programming for covering radius of orthogonal arrays}

Assume that $f$ is a polynomial of degree at most $\tau$ which is non-positive on the interval $[-1,s]$ but still has positive Krawtchouk coefficient $f_0>0$. If the covering radius of an $OA(n,M,\tau)$ is less than $s$, then \eqref{dd-sys}, applied for that OA and a point $x$ where the covering radius is met, would give a contradiction. Therefore, the covering radius of that OA is at least $s$ and we have the following general linear programming theorem (see Section 2 in \cite{FL} for this form of exposition).

\begin{theorem} \cite{Tie90,Tie91,FL} 
\label{th:tie} 	
Let $\tau$ be a positive integer, $s \in [-1,1)$ be a real number, and the polynomial $f \in \mathbb{R}[t]$ satisfy 

{\rm (A1)} $f(t) \leq 0 \ \forall t \in [-1,s]$;

{\rm (A2)} $\deg(f) \leq \tau$;

{\rm (A3)} $f_0 > 0$ in the Krawtchouk expansion $f(t)=\sum_{i=0}^{\tau} f_iQ_i(t)$.

\noindent
Then the covering radius of any binary $OA(n,M,\tau)$ satisfies
$\rho(C) \geq s$.
\end{theorem}

The next step is the maximization of $s$ for given $n$ and $\tau$ (note that $M$ is not involved). 
The resulting bound is (\ref{FL_bound}) which was proved first by Tiet{\"a}v{\"a}inen in 1990. For even $\tau=2k$ it is attained when $C$ is a tight (binary) OA of strength $2k$ (see the examples on pages 283-284 in \cite{FL}).

\begin{table}
\begin{center}
\caption{ Fazekas--Levenshtein lower bounds (Tiet{\"a}v{\"a}inen upper bounds) for the covering radius of binary orthogonal arrays in $F_2^n$ of strength $1 \leq \tau \leq 8$. }
\label{tab:1}
\begin{tabular}{|c|c|c|c|c|}
  \hline\noalign{\smallskip}
  Strength & Lower bound & Upper bound  \\
    &  & (in terms of distances)  \\
  $\tau = 2k-1+e$ & $\rho(C) \geq t_k^{0,e}$ & $R(C) \leq  d_k^{(n-e)}$ \\
  \noalign{\smallskip}\hline\noalign{\smallskip}  
  1 &  $t_1^{0,0} = 0$ &  $d_1^{(n)} = \frac{n}{2}$ \\
  \noalign{\smallskip}\hline\noalign{\smallskip}  
  2 &  $t_1^{0,1} = \frac{1}{n}$ & $d_1^{(n-1)} = \frac{n-1}{2}$ \\
  \noalign{\smallskip}\hline\noalign{\smallskip}  
  3 &  $t_2^{0,0} = \frac{1}{\sqrt{n}}$ & $d_2^{(n)} = \frac{n-\sqrt{n}}{2}$ \\
  \noalign{\smallskip}\hline\noalign{\smallskip}  
  4 &  $t_2^{0,1} = \frac{\sqrt{n-1}+1}{n}$ & $d_2^{(n-1)} = \frac{n-1-\sqrt{n-1}}{2}$ \\
  \noalign{\smallskip}\hline\noalign{\smallskip}
  5 &  $t_3^{0,0} = \frac{\sqrt{3n-2}}{n}$ & $d_3^{(n)} = \frac{n-\sqrt{3n-2}}{2}$ \\
  \noalign{\smallskip}\hline\noalign{\smallskip}
  6 &  $t_3^{0,1} = \frac{\sqrt{3n-5}+1}{n}$ & $d_3^{(n-1)} = \frac{n-1-\sqrt{3n-5}}{2}$ \\
  \noalign{\smallskip}\hline\noalign{\smallskip}
  7 &  $t_4^{0,0} = \frac{\sqrt{3n-4+\sqrt{6n^2-18n+16}}}{n}$ & $d_4^{(n)} = \frac{n-\sqrt{3n-4+\sqrt{6n^2-18n+16}}}{2}$ \\
  \noalign{\smallskip}\hline\noalign{\smallskip}
  8 &  $t_4^{0,1} = \frac{\sqrt{3n-7+\sqrt{6n^2-30n+40}}+1}{n}$ & $d_4^{(n-1)} = \frac{n-1-\sqrt{3n-7+\sqrt{6n^2-30n+40}}}{2}$ \\  
  \hline\noalign{\smallskip}
\end{tabular}
\end{center}
\end{table}

\subsection{Adjacent polynomials}

Levenshtein (cf. \cite[Section3]{Lev92}, \cite[Sections 3 and 6]{Lev98} introduced the so-called adjacent (to $Q_i(t)$) polynomials denoted by $Q_i^{a,b}(t)$, where $a,b \in \{0,1\}$, to serve in the linear programming framework (note that $Q_i^{0,0}(t)=Q_i(t)$). For example, the polynomials $Q_i^{1,0}(t)$ and
$Q_i^{1,1}(t)$ were used in the formulation, proof, and final representation of the Levenshtein universal bounds for maximal cardinality of codes with prescribed length and minimum distance \cite{Lev92,Lev95,Lev98} and, similarly, the polynomials $Q_i^{0,1}(t)$ were utilized for the universal bounds for designs (again \cite{Lev92,Lev95,Lev98}) and their covering radius (see \cite{FL}).

We shall need the polynomials $Q_i^{0,1}(t)$ to explain the Fazekas-Levenshtein bound and their extensions 
$Q_i^{0,\ell}(t)$ to obtain and explain our bound. We notice the explicit relation \cite{FL,Lev98}
\[ Q_i^{0,1}(t) = \frac{K_i^{(n-1)}(n(1-t)/2)}{{n-1 \choose i}} \]
between the $(0,1)$-polynomials and the Krawtchouk polynomials. We note that with $\ell=-1$ as the boundary case, in our notation the polynomials $Q_i^{0,1}(t)$ would be $Q_i^{0,-1}(t)$.

\subsection{Positive definite signed measures}

In this and the following subsections we develop a technique to be applied for a deeper investigation of the covering radius of orthogonal arrays relying on analysis of the geometry around the points which are antipodal to points where the covering radius is met. We will introduce certain signed measure and will derive their properties which will give us the necessary tools. 

\begin{definition} 
A signed Borel measure $\mu$ on $\mathbb{R}$ for which all polynomials are integrable
is called {\it positive definite up to degree $m$} if for all real polynomials $p \not\equiv 0$
of degree at most $m$ we have $\int p^2 (t) d \mu(t) > 0$.  For such $\mu$, the bi-linear form
\begin{equation}\label{ipsm}
\langle f, g\rangle_{\mu}:=\int f(t)g(t)\, d\mu(t), \nn
\end{equation}
is an inner product on the space $\mathcal{P}_m$.
\end{definition} 

Signed (discrete) measures for linear programming in $F_q^n$ were considered in \cite{BDHSS-ieee}, where it was proved that the signed measures
\[ d\nu_\ell(t) := (t-\ell)(1-t)d\mu_n (t),\ \ 
d\nu_s(t) := (s-t)(1-t)d\mu_n(t) \]
under certain natural conditions for the parameters $s$ and $\ell$ are positive definite up to degree $k-1$. Here we need a simpler, in a sense, signed measure, which was not covered in \cite{BDHSS-ieee}.

Given $\ell \in [-1,0)$ we define the signed measure on $[-1,1]$
\[ d \mu_n^{(\ell)}(t) := c_{n,\ell} (t-\ell) d \mu_n(t), \ \ \ t \in [-1,1], \]
where $c_{n,\ell} := - 1/\ell>0$ is a normalizing constant.
We shall prove that $d \mu_n^{(\ell)}(t)$ is positive definite up to 
degree $k-1$ under certain natural assumption. The analogous property of the Euclidean counterpart of $d \mu_n^{(\ell)}(t)$ was proved in \cite[Lemma 2.2]{BDHSS-DCC2019}. 

Let
\begin{equation} \label{roots-q} 
t_{k,1}< t_{k,2}< \cdots < t_{k,k} 
\end{equation}
be the roots of the Krawtchouk $Q_k(t)$ of degree $k$, ordered increasingly. Note that $-1<t_{k,1}$ and $t_{k,k}<1$ as usual in the theory of orthogonal polynomials. 

\begin{lemma} \label{lem_pos_def} For given $k>1$, let $\ell$  satisfy $\ell < t_{k,1}$. Then the measure $d\mu_{n}^{(\ell)} (t)$ is positive definite up to degree $k-1$. \end{lemma}

{\it Proof.}
We employ a quadrature formula defined on the set of the roots \eqref{roots-q} of the polynomial $Q_k(t)$ as nodes. Using the associated Lagrange basis polynomials 
\begin{equation} \label{L-poly}
L_i:=\prod_{j \neq i} (t-t_{k,j}), \ i=1,2,\dots, k,
\end{equation}
we define the corresponding formula weights 
\[ \rho_i:= \int_{-1}^1 L_i (t) d\mu_n(t),  \ i=1,2, \dots, k. \]
Then we verify that the quadrature formula
\begin{equation}\label{v_quadrature}
f_0:= \int_{-1}^1 f(t) d\mu(t)= \sum_{i=1}^k \rho_i f(t_{k,i} )
\end{equation}
is exact for all polynomials of degree up to $2k-1$.

It is immediate that \eqref{v_quadrature} holds for polynomials of degree at most $k-1$ since the Lagrange polynomials \eqref{L-poly} form a basis of the space of polynomials of degree at most $k-1$. For polynomials $f$ of degree $k \leq \deg(f) \leq 2k-1$ we use the unique representation 
\[ f(t)=q(t)Q_k(t)+r(t)=q(t)\prod_{i=1}^k (t-t_{k,i}) +r(t),\]
where $\deg(r) \leq k-1$, to see that \eqref{v_quadrature} holds for $f$. 

Using \eqref{v_quadrature} for the square of any polynomial $p(t)$ of degree at most $k-1$, we obtain
\begin{eqnarray*}
\int_{-1}^1 p^2(t) d\mu_n^{\ell} (t) = \int_{-1}^1 p^2(t)(t-\ell) d\mu_n(t)
= \sum_{i=1}^k \rho_i p^2(t_{k,i} )(t_{k,i} -\ell) \geq 0,
\end{eqnarray*}
with equality if and only if $p(t)\equiv 0$. Therefore, the measure $d\mu_n^{\ell} (t)$ is positive definite up to degree $k-1$. \hfill $\Box$

\section{Properties of the polynomials $Q_i^{0,\ell} (t)$}

Lemma \ref{lem_pos_def} implies the existence (via the Gramm-Schmidt ortogonalization procedure) of a finite sequence of polynomials
$(Q_i^{0,\ell}(t))_{i=0}^k$ which are orthogonal with respect to the measure $d \mu_n^{\ell}(t)$. Moreover, with normalization $Q_i^{0,\ell}(1)=1$ these polynomials are uniquely determined. 
This allows us to find explicit formulas
the polynomials $Q_i^{0,\ell}(t)$ and to investigate their main properties. 

We derive the necessary properties of the series 
$\{Q_i^{0,\ell}(t)\}_{i=0}^{k}$. The results and proofs are parallel to the corresponding results and proofs in the Euclidean case \cite{BS21}. 

\subsection{Explicit formula for the polynomials $Q_i^{0,\ell}$}

Consider the Christoffel-Darboux kernel that corresponds to the Krawtchouk polynomials (cf. \cite{Sze})
\[ T_i(u,v) := \sum_{j=0}^i r_j Q_j(u) Q_j(v). \]

\begin{theorem}
\label{explicit-q}
Let $\ell$ and $k$ be such that $t_{k+1,1}<\ell<t_{k,1}$. Then
\begin{equation} \label{Pi0-ell}
Q_i^{0,\ell}(t) = \frac{T_i (t,\ell)}{T_i(1,\ell)} = m_i^{0,\ell} t^i + \cdots,\quad i=0,1,\dots, k,
\end{equation}
with $m_i^{0,\ell}>0$. 
\end{theorem}

{\it Proof.} It follows from the Christoffel-Darboux formula 
\begin{equation} \label{Ch-D}
\frac{T_i (t,\ell)}{T_i(1,\ell)} = \frac{(1-\ell)\left( Q_{i+1} (t) - Q_i(t)Q_{i+1}(\ell)/Q_{i}(\ell) \right)}
{(t-\ell)\left(1-Q_{i+1}(\ell)/Q_{i}(\ell)\right)}
\end{equation}
that the polynomial $(t-\ell)T_i(t,\ell)$ is
a linear combination of the polynomials $Q_{i+1}(t)$ and  $Q_i(t)$ for every $i \geq 0$. This immediately implies that the polynomial 
$T_i(t,\ell)$ (of degree $i$) itself is orthogonal to any polynomial of degree at most $i-1$ with respect to the measure $d\mu_{\ell}(t)$. Now
\eqref{Pi0-ell} follows from the positive definiteness of $d\mu_n^{(\ell)}(t)$ up to degree $k-1$,
the uniqueness of the Gram-Schmidt orthogonalization process and the normalization. The comparison of coefficients
in \eqref{Pi0-ell} shows that $m_i^{0,\ell}>0$, $i=0,1,\ldots,k$.
\hfill $\Box$

The boundary case $\ell=-1$ leads to the polynomials $Q_i^{0,1}(t)$, which are important ingredients in the Fazekas-Levenshtein framework. In this case the formula \eqref{Ch-D} coincides with (5.65) from \cite{Lev98}. 

We conclude this subsection with expicit formulas for the polynomials $Q_1^{0,\ell}$ and $Q_2^{0,\ell}$ that will be used later. We have 
\begin{equation} \label{Ex3}
   Q_1^{0,\ell}(t) = \frac{1+n\ell t}{1+n\ell}, \ \ 
   Q_2^{0,\ell}(t) = \frac{n^2(n\ell^2-1)t^2 + 2n\ell(n-1)t - n^2\ell^2 +3n-2}{(n-1)((n\ell+1)^2-(n-1))}.
\end{equation}

\subsection{Interlacing of roots}

The explicit formula \eqref{Ch-D} relates the polynomials 
$Q_i^{0,\ell}(t)$ and $Q_{i+1}(t)$ and provides a tool to derive the interlacing properties of the zeros of $Q_i^{0,\ell}(t)$ with respect to the zeros of its "neighbours"\ $Q_{i+1}(t)$ and $Q_i(i)$.

We denote by
\begin{equation} \label{roots-0ell}
    t_{i,1}^{0,\ell} < t_{i,2}^{0,\ell} < \cdots < t_{i,i}^{0,\ell} 
\end{equation}
the zeros of $Q_i^{0,\ell}(t)$ and recall \eqref{roots-q}. In fact, we are mostly interested in the zeros of the highest degree polynomials in the series. 

\begin{theorem}
\label{roots-of-p0-ell}
Let $\ell$ and $k$ be such that $t_{k+1,1}<\ell<t_{k,1}$ and
$Q_{k+1}(\ell)/Q_{k}(\ell)<1$. Then
the zeros \eqref{roots-0ell} of $Q_i^{0,\ell}(t)$ belong to 
the interval $(\ell,1)$ and the interlacing rules
\begin{equation}\label{Interlacing}\begin{split}
t_{i,j}^{0,\ell} &\in (t_{i,j}, t_{i+1,j+1}), \ i=1,\dots, k-1, j=1,\dots, i ;\\
t_{k,j}^{0,\ell} &\in (t_{k+1,j+1}, t_{k,j+1}), \ j=1,\dots,k-1, \ t_{k,k}^{0,\ell}\in (t_{k+1,k+1},1),
\end{split}
\end{equation}
hold.
\end{theorem}

{\it Proof.} It follows from \eqref{Pi0-ell} and \eqref{Ch-D} that the zeros of $Q_k^{0,\ell}(t)$ are exactly the solutions of the equation
\begin{equation} \label{FracEq}
\frac{Q_{i+1}(t)}{Q_i(t)} = \frac{Q_{i+1}(\ell)}{Q_i(\ell)}.
\end{equation}
We prove \eqref{Interlacing} by analyzing \eqref{FracEq}. There are two essentially different situations: the cases $i<k$ and $i=k$. 

Let $i<k$. From the general properties of orthogonal polynomials (cf. \cite[Chapter 4]{Sze}) the zeros of $Q_{i+1}(t)$ and $Q_i(t)$ are interlaced and belong to the interval $[t_{k,1},t_{k,k}]$. The inequality $\ell <t_{1,k}$ implies that
${\rm sign}\,Q_i(\ell)=(-1)^i$. Therefore, the right hand side of \eqref{FracEq} is equal to a negative constant. The rational function $Q_{i+1}(t)/Q_i(t)$ in the left hand side has simple poles at the zeros $t_{i, j}$, $j=1,\dots, i$, of $Q_i(t)$ and simple zeros at the zeros of $t_{i+1,j}$, $j=1,\dots, i+1$, of $Q_{i+1}(t)$. This means that there is at least one solution $t_{i,j}^{0,\ell}$ of \eqref{FracEq} in every subinterval $(t_{i,j}, t_{i+1,j+1})$, $j=1,\dots,i$, (in particular, $\ell<t_{k,1}<t_{k,1}^{0,\ell}$) which in fact accounts for all zeros of $Q_i^{0,\ell}(t)$ and we are done in this case.

In the case $i=k$ it follows from the inequalities $t_{k+1,1}<\ell<t_{k,1}$ that $Q_{k+1}(\ell)/Q_k(\ell)>0$ and we account similarly to above for the first $k-1$ solutions of \eqref{FracEq}, i.e. we have \[t_{k,j}^{0,\ell} \in (t_{k+1,j+1}, t_{k,j+1}), \ j=1,\dots,k-1.\]
For the largest zero of $Q_{k}^{0,\ell}(t)$ we use that $Q_{k+1}(t)/Q_k(t)>0$ for every $t>t_{k+1,k+1}$ and the last ration tends to infinity as $t$ tends to infinity. Thus, we find at least one more solution of \eqref{FracEq} which is greater than $t_{k+1,k+1}$. This accounts for $t_{k,k}^{0,\ell}$, the last zero of $Q_{k}^{0,\ell}(t)$. Moreover, since $Q_{k+1}(\ell)/Q_{k}(\ell)<1=Q_{k+1}(1)/Q_k(1)$ by assumption, we conclude that $t_{k,k}^{0,\ell}<1$. This completes the proof.
\hfill $\Box$

\begin{remark} We remarks that the condition $Q_{k+1}(\ell)/Q_{k}(\ell)<1$ was used only to ensure $t_{k,k}^{0,\ell}<1$. Without that condition but still with 
$t_{k+1,1}<\ell<t_{k,1}$ the polynomials $Q_i^{0,\ell}$ are well defined and \eqref{Interlacing} follows except for $t_{k,k}^{0,\ell}<1$. 
\end{remark}

A three-term recurrence relation for the polynomials $Q_i^{0,\ell}$, $i \geq 0$, can be derived by standard means, but we do not need it in this paper. Instead, we turn to the quadrature formula which is provided by the zeros of $P_{k}^{0,\ell}(t)$.  

\subsection{A quadrature formula}

We denote by $L_i (t)$, $i=0,1,\ldots,k$, the Lagrange basic polynomials corresponding to the set of nodes 
\begin{equation} \label{nodes-ell}
\ell<t_{k,1}^{0,\ell}<t_{k,2}^{0,\ell}<\cdots<t_{k,k}^{0,\ell},
\end{equation}
where $L_0(t)$ corresponds to $\ell$ (i.e., it has the same zeros as $Q_k^{0,\ell}$) and $L_i(t)$ corresponds to $t_i^{0,\ell}$ for $i=1,\ldots,k$. Further, we denote by
\[ \theta_i:=\int_{-1}^1 L_i(t) d\mu(t), \ \ i=0,1,\dots,k, \]
the weights of our quadrature formula. 

 \begin{theorem} \label{QFtheorem} The quadrature formula
\begin{equation}
\label{QF}
f_0 = \int_{-1}^1 f(t) d \mu(t)
    =\theta_0 f(\ell)+ \sum_{i=1}^{k} \theta_i f(t_{k,i}^{0,\ell})
\end{equation}
is exact for all polynomials of degree at most $2k$ and its  weights $\theta_i$, $i=0,1,\dots, k$, are positive.
\end{theorem}

{\it Proof.}  We argue as in Lemma \ref{lem_pos_def}. The formula \eqref{QF} is exact for the Lagrange basis at $k+1$ nodes
from \eqref{nodes-ell} and hence for all polynomials of degree at most $k$. 

Given a polynomial $f(t)$ of degree $k+1 \leq \deg(f) \leq 2k$, we write it as
\begin{equation} \label{div}
f(t)=(t-\ell)Q_{k}^{0,\ell} (t)q(t)+r(t),
\end{equation}
where the quotient $q(t)$ has degree at most $k-1$ and the remainder $r(t)$ has degree at most $k$.
We again integrate over $[-1,1]$ with respect to $\mu(t)$ and use the orthogonality of $P_{k}^{0,\ell}(t)$ to all polynomials of degree at
most $k-1$ with respect to $d\mu_{\ell}(t)=c_{n,\ell} (t-\ell)d\mu(t)$ to see that
\[ f_0=r_0=\theta_0 r(\ell)+ \sum_{i=1}^{k} \theta_i r(t_{k,i}^{0,\ell})=\theta_0 f(\ell)+ \sum_{i=1}^{k} \theta_i f(t_{k,i}^{0,\ell}) \]
by \eqref{div} (here $r_0$ is the zeroth Krawtchouk coefficient of $r$). Therefore, \eqref{QF} holds for $f$ which completes the proof of the exactness.

Next, we show the positivity of the weights $\theta_i$, $i=0,\dots, k$, by using suitable polynomials in \eqref{QF}. First, we fix $i \in \{1,2,\ldots,k\}$ and apply $f(t)=(t-\ell)\left( u_{i} (t)\right)^2$ in \eqref{QF}, where
\[ u_{i}(t)=\frac{P_{k}^{0,\ell}(t)}{t-t_{k,i}^{0,\ell}}.  \]
Since $\deg(u_i)=k-1$, we have $\deg(f)=2k-1$ and the formula \eqref{QF}) is exact for $f$. This gives
\[
f_0=\theta_i (t_{k,i}^{0,\ell}-\ell) \left(u_{i} (t_{k,i}^{0,\ell})\right)^2,
\]
whence it follows that the sign of $\theta_i$ is the same as the sign of $f_0$. Now, from
\[ f_0=\int_{-1}^1 (t-\ell)\left( u_{i} (t)\right)^2 d \mu(t)= \frac{1}{c_{n,\ell}} \int \left( u_{i} (t)\right)^2 d \mu_\ell(t) >0,\]
where we use the fact that $\mu_\ell(t)$ is positive definite up to degree $k-1=\deg(u_i)$, we conclude that $\theta_i>0$, $i=1,\dots, k$.
For $\theta_0$, we use $f(t)=\left( P_k^{0,\ell}(t)\right)^2$ of degree $2k$ in \eqref{QF} to obtain $\theta_0 f(\ell)=f_0$ to see that the
sign of $\theta_0$ is the same as the sign of $f_0$. Obviously $f_0>0$ and this completes the proof. \hfill $\Box$

\section{Improving Fazekas-Levenshtein bounds}

\subsection{Rao bound}

Denote by $B(n,\tau)$ the mimimum cardinality $M$ of an OA$(M,n,q,\tau)$ for fixed length $n$, alphabet size $q$, and strength $\tau$.
The following bound was first proved by Rao \cite{Rao} in 1947:
\begin{equation}
\label{R-bound}
B(n,\tau) \geq R(n,\tau) := q^{1-\varepsilon} \sum_{i=0}^{k-1+\varepsilon} {n-1+\varepsilon \choose i} (q-1)^i,
\end{equation}
where $\tau = 2k-1+\varepsilon$, $\varepsilon \in \{0,1\}$ just indicates the parity of $\tau$. For $q=2$, \eqref{R-bound} becomes
\begin{equation}
\label{R-bound-2}
B(n,\tau) \geq R(n,\tau) := 2^{1-\varepsilon} \sum_{i=0}^{k-1+\varepsilon} {n-1+\varepsilon \choose i}.
\end{equation}
The first few bounds \eqref{R-bound-2} are 
\[ R(n,1)=2, \ R(n,2)=n+1, \ R(n,3)=2n, \ R(n,4)=\frac{n^2+n+2} {2}, \] 
\[ R(n,5)=n^2-n+2, R(n,6)=\frac{(n+1)(n^2-n+6)}{6} \]
(see \cite{Del-dis,Del73,Lev98}). The OAs attaining the Rao bound are called tight and exist rarely (see, e,g,, \cite{Ban21,GSV,GS26,MK94,N79}). We are interested in OAs of strength $2k$ which have cardinality slightly larger than $R(n,2k)$.  

\subsection{Notations for the structure of OAs}

We are interested in the structure of OAs whose cardinality is close to the Rao bound. More precisely,
we are interested in linear programming bounds which follow from using Theorem \ref{th:pi-f} with suitable polynomials. 
Apart from the interest here, we believe that our information could be useful for proving nonexistence
results (see, e.g., \cite{BMS}).

Let $C \subset F_2^n$ be an OA. For any point $y \in F_2^n$, we define the (multi)set
\[ I(y) = \{ \langle x,y \rangle : x \in C\} = \{ t_1(y),t_2(y),\ldots,t_{|C|}(y) \}, \]
where we order $I(y)$ by \[ -1 \leq t_1(y) \leq t_2(y) \leq \cdots \leq t_{|C|}(y) \leq 1 \]
(note that $t_{|C|}(y)=1 \iff y \in C$).

Then \eqref{dd-sys} can be written as
\begin{equation} \label{def-des-point-wise}
\sum_{j=1}^{|C|} f(t_j(y))=f_0|C|
\end{equation}
to hold for every point $y \in F_2^n$, every orthogonal array $C \subset F_2^n$ of strength $\tau$, and every real polynomial $f$ (with its Krawtchouk expansion \eqref{kr-exp}) of degree at most $\tau$. 

Denote by 
\[ D(C):=\{ y \in F_2^n : t_{|C|}(y)=\rho(C) \} \]
the set of points $y \in F_2^n$, where the covering radius is realized. 

Denote by $\overline{C}$ the set of points which are antipodal to points of $C$, i.e.
\[ \overline{C}:=\{ u\in F_2^n : \exists x \in C \mbox{ such that } d(x,u)=n \iff t_1(x)=\langle x,u \rangle =-1\}. \]
Note that the set of antipodal points is never empty and if $u \in \overline{C}$, then its distance $d(u,C)$ to $C$ is strictly less than $n$ whenever $|C|>1$. Note also that $C=\overline{C}$ if and only if $C$ is antipodal. 

\subsection{Binary OAs of strength $2k$}

The parameter $k$ comes henceforth from $\tau=2k$, the strength of the orthogonal arrays under consideration. 
Let $C \subset F_2^n$ be a binary OA of strength $2k$ and cardinality $|C|>R(n,2k)$ 
and let $y \in F_2^n$ is a point 
which realizes the covering radius of $C$, i.e. $y \in D(C)$. 
 
The tight binary OAs of strength $2k$ have $t_1(y)=-1$ for every $y \in D(C)$ \cite[Theorem 3]{FL}. In other words, $-y \in \overline{C}$ for every $y \in D(C)$ for the tight binary OAs of strength $2k$. 

Therefore, it is natural to investigate how close is $t_1(y)$ to $-1$ (for a point $y \in D(C)$) at least for cardinalities which are close to the Rao bound $R(n,2k)$.
For chosen $\ell$, which is close to $-1$, we consider two cases: $t_1(y) \in [-1,\ell]$ for all $y \in D(C)$ and $t_1(y) \geq \ell$ for some $y \in D(C)$.

\subsubsection{Case $-1 < \ell \leq t_1(y)$ for some $y \in D(C)$}

We start with a general observation that is, in fact, a modification of Theorem \ref{th:tie} for this case. We need to adjust only condition (A1). 

\begin{theorem}
\label{th:tie-modified} 	
Let $\tau$ be a positive integer, $\ell \in (-1,0)$, $s \in [\ell,1)$, and the polynomial $f \in \mathbb{R}[t]$ satisfy 

{\rm (A1)$^\prime$} $f(t) \leq 0 \ \forall t \in [\ell,s]$;

{\rm (A2)} $\deg(f) \leq \tau$;

{\rm (A3)} $f_0 > 0$ in the Krawtchouk expansion $f(t)=\sum_{i=0}^{\tau} f_iQ_i(t)$.

\noindent
Assume that $C$ is a binary $OA(n,M,\tau)$ such that there exists $y \in D(C)$ with $\ell \leq t_1(y)$. Then the covering radius of $C$ satisfies $\rho(C) \geq s$.
\end{theorem}

\begin{proof}
 Let $C$ be a binary $OA(n,M,\tau)$ and the point $y \in D(C)$ be such that $t_1(y) \geq \ell$. Let $f$ satisfy the conditions (A1)$^\prime$, (A2), and (A3).
 Suppose for a contradiction that $\rho(C) < s$. Applying \eqref{dd-sys} (or, equivalently, \eqref{def-des-point-wise}) with $f$, $C$, and $y$, we obtain that the left-hand side is non-positive due to condition (A1)$^\prime$ and inequality $\rho(C) < s$. On the other hand, the right-hand side $f_0|C|$ is positive due to (A3). This contradiction completes the proof. 
\end{proof}

We are now in a position to improve the Fazekas-Levenshtein bound in the case $-1<\ell \leq t_1(y)$. We apply Theorem \ref{th:tie-modified} with suitable polynomial and use the quadrature formula \eqref{QF}. 

\begin{theorem} \label{impr_fl}
Let $C \subset F_2^n$ be an OA of strength $2k$. Suppose that $\ell$ is such that $t_{k+1,1}<\ell<t_{k,1}$,
$Q_{k+1}(\ell)/Q_{k}(\ell)<1$, and $-1<\ell \leq t_1(y)$
for some $y \in D(C)$. Then 
\begin{equation} \label{main-bound}
    \rho(C) \geq t_{k,k}^{0,\ell}.
\end{equation}
\end{theorem}

\begin{proof}
We apply Theorem \ref{th:tie-modified} with the polynomial \[ f(t)=(t-\ell)(t-t_{k,k}^{0,\ell}+\varepsilon) \prod_{i=1}^{k-1} (t-t_{k,i}^{0,\ell})^2, \]
where $\varepsilon>0$ is a small positive number. It is easy to see that the conditions of Theorem \ref{th:tie-modified} are satisfied for $f$ with $s=t_{k,k}^{0,\ell}-\varepsilon$. Indeed, $f(t) \leq 0$ for every $t \in [\ell,t_{k,k}^{0,\ell}-\varepsilon]$ for (A1)$^\prime$, $\deg (f)=2k=\tau$ for (A2), and for (A3) by the quadrature formula \eqref{QF} for $f$ we compute
\begin{eqnarray*} 
f_0 &=& \theta_0 f(\ell)+ \sum_{i=1}^{k} \theta_i f(t_{k,i}^{0,\ell})= \theta_kf(t_{k,k}^{0,\ell}) \\
&=& \varepsilon \theta_k (t_{k,k}^{0,\ell}-\ell)  \prod_{i=1}^{k-1} (t_{k,k}^{0,\ell}-t_{k,i}^{0,\ell})^2>0. \end{eqnarray*}
Therefore, Theorem \ref{th:tie-modified} implies that  $\rho(C) \geq t_{k,k}^{0,\ell}-\varepsilon$. Since this is true for every small enough $\varepsilon>0$, we conclude that $\rho(C) \geq t_{k,k}^{0,\ell}$ as required. 
\end{proof}

If the boundary case $\ell=-1$ is allowed in Theorem \ref{impr_fl}, we will obtain the Fazekas-Levenshtein bound $\rho(C) \geq t_k^{0,1}$. In other words, in our notations the Fazekas-Levenshtein bound is $\rho(C) \geq t_k^{0,-1}$.
Therefore, we have improved that bound whenever there exists an $\ell$ and a point $y \in D(C)$ such that $-1<\ell \leq t_1(y)$. This also follows from more general results on zeros interlacing from \cite[Chapter 5.3]{Lev98}; we have 
\[ t_{k,k}^{0,\ell} \geq t_{k+1,k+1} >t_k^{0,1} \]
from our Theorem \ref{roots-of-p0-ell} and Lemma 5.30 from \cite{Lev98}. 

Reformulated in terms of distances, the bound \eqref{main-bound} says
\[ R(C) \leq d_k^{0,\ell}:=\frac{n(1-t_{k,k}^{0,\ell})}{2}. \]

\begin{example}
   Utilizing the explicit formulas \eqref{Ex3} we derive the first two (i.e., for $k=1$ and 2) bounds from Theorem \ref{impr_fl} as follows:
\[ \rho(C) \geq t_{1,1}^{0,\ell} = - \frac{1}{n\ell}>t_1^{0,1}=\frac{1}{n}, \]
for $\tau=2$ and all $\ell \in (-1,t_{1,1}=0)$, and 
\[ \rho(C) \geq t_{2,2}^{0,\ell} = \frac{-(n-1)\ell+\sqrt{n(n\ell^2-1)(n\ell^2-2) +\ell^2 +n-2}}{n(n\ell^2-1)}>t_2^{0,1}=\frac{\sqrt{n-1}+1}{n},  \]
for $\tau=4$ and for all $\ell \in (-1, t_{2,1}=-1/\sqrt{n})$. The condition $Q_{k+1}(\ell)/Q_{k}(\ell)<1$ for $k=1$ and 2 holds for $\ell<-1/n$ and $\ell<-1/\sqrt{n}$, respectively. Indeed, we have to consider
\[ \frac{Q_2(\ell)}{Q_1(\ell)}=\frac{n\ell^2-1}{(n-1)\ell}<1, \ \frac{Q_3(\ell)}{Q_2(\ell)}=\frac{(n\ell^2-3n+2)\ell}{(n-2)(n\ell^2-1)}<1 \]
and straightforward calculations give the desired. 

In terms of distances and the same conditions for $\ell$, we have 
\[ R(C) \leq d_1^{0,\ell} = \frac{n\ell + 1}{2\ell}<d_1^{0,1}=\frac{n-1}{2}, \ \]
for $\tau=2$, and 
\[ R(C) \leq d_2^{0,\ell} = \frac{n}{2}+\frac{(n-1)\ell-\sqrt{n(n\ell^2-1)(n\ell^2-2) +\ell^2 +n-2}}{2(n\ell^2-1)}< d_2^{0,1} = \frac{n-1-\sqrt{n-1}}{2},  \]
for $\tau=4$.
%All inequalities $t_{i,i}^{0,\ell}>t_i^{0,1}$, $i=1,2$, and $d_i^{0,\ell}<d_i^{0,1}$, $i=1,2$, respectively, hold (but are not necessarily valid bounds) for all $\ell \in (-1,0)$. \textcolor{blue}{Is this true for $k=2$? Maya said No, 29.8.25).} 
\end{example}

In practice, the bound \eqref{main-bound} should be computed as follows. Given $n$, we assume that the Krawtchouk polynomials 
$\{Q_i(t)\}_{i=0}^n$ are constructed and their roots are collected in a database. Next, for given $k$, we choose $\ell=t_{k,1}-\mu$ for some small $\mu>0$. Then we construct the family $\{Q_i^{0,\ell}(t)\}_{i=0}^k$ or directly compute the polynomial $Q_k^{0,\ell}(t)$ from the formula \eqref{Pi0-ell}. Its largest root is the desired bound.

\subsection{Case $t_1(y)=-1$ for every $y \in D(C)$, i.e. $D(C) \subseteq \overline{C}$}

\subsubsection{A general bound}
We apply another modification of the linear programming method that utilizes the equalities $t_1(y)=-1$ and $t_{|C|}(y)=\rho(C)$ for $y \in D(C)$. 

\begin{theorem} \label{case2-thm}
Let $C \subset F_2^n$ be an OA of strength $2k$ such that $D(C) \subseteq \overline{C}$. Let 
\[ f(t)=\left(t+1-\frac{2}{n}\right)(t-\rho(C))A^2(t), \]
where $A(t)$ is a polynomial of degree $k-1$. Then, with $y \in D(C)$, we have $f_0|C|-f(-1) \leq 0$ and, consequently, 
\begin{equation} \label{q}
    \rho(C) \geq q_k(n,|C|),
\end{equation}
where the last two inequalities are equivalent. 
\end{theorem}

\begin{proof}
We consider \eqref{def-des-point-wise} for $f$, $C$ and a point $y \in D(C)$. Since $f(t_i(y)) \leq 0$ for $2 \leq i \leq |C|$, we conclude that 
\begin{equation} \label{t1=-1}
    f_0|C|-f(-1)= \sum_{i=2}^{|C|} f(t_i(y)) \leq 0,
\end{equation}
whence the inequality \eqref{q} for $\rho(C)$ is derived by solving the inequality $f_0|C|-f(-1) \leq 0$ with respect to $\rho(C)$. 
\end{proof}

\begin{remark}
Note that the bound \eqref{q} depends on the cardinality of the OA while both the bound \eqref{main-bound} and the Fazekas-Levenshtein bound depend only on the strength. When the polynomial $A$ varies, $q_k(n,|C|)$ becomes a functional of $A$.  
\end{remark}

One needs to maximize the functional $q_k(n,|C|)$  over the polynomial $A(t)$ in Theorem \ref{case2-thm} in order to get the best bound \eqref{q}. This seems to be easy numerically but it is quite difficult to find analytical expressions for $k>1$. 

\begin{example}
    For $k=1$. we have $A(t)=1$, $f(-1)=2(1+\rho(C))/n$, and $f_0=1/n-\rho(C)(1-2/n)$. Solving $f_0|C|-f(-1) \leq 0$, we obtain
\begin{equation} \label{k=1t=-1}
\rho(C) \geq q_1(n,|C|)=\frac{|C|-2}{(n-2)|C|+2}.
\end{equation}
Note that for $|C|=n+1=R(n,2)$, the case of tight orthogonal arrays of strength 2, the bound \eqref{k=1t=-1} coincides with the Fazekas-Levenshtein bound $\rho(C) \geq 1/n$. It is attained whenever a Hadamard matrix of order $n+1$ exists (cf. Theorem 7.3 in \cite{HSS}). For $|C|>n+1$, \eqref{k=1t=-1} is better than the Fazekas-Levenshtein bound. \end{example}

\begin{example} 
For $k=2$ we have to consider 
\[ f(t)=\left(t+1-\frac{2}{n}\right)(t-\rho(C))(t-a)^2, \]
where $a$ will be optimized for the best results via \eqref{t1=-1}. 
Since 
\[ f_0=\frac{1}{n^3}\left(n^2a^2-2n(n-2)a+3n-2-\rho(C)\left(n^2(n-2)a^2-2n^2a+n(n-2)\right)\right), \]
we derive from \eqref{t1=-1} that
\[ \rho(C) \geq q_2(n,|C|)=\frac{|C|(n^2a^2-2n(n-2)a+3n-2)-2n^2(1+a)^2}{n|C|(n(n-2)a^2-2na+n-2)+2n^2(1+a)^2} \]
for $n \geq 4$ (it is easy to check that the denominator is positive for $n \geq 4$). For $n=5$ and $|C|=16$ (this is the even-weight code of length 5 that is the only tight orthogonal array of strength 4, cf. \cite[Theorem 1.2]{GSV}) we obtain
\[ q_2(5,16)=\frac{175a^2-290a+79}{5(125a^2-70a+29)}, \]
which is maximized for $a=-1/5$, giving the bound $\rho(C) \geq 3/5$ (in distances, $R(C) \leq 1$). This coincides with the Fazekas-Levenshtein bound and is attained by the said code. 

In general, differentiating $q_2(n,|C|)$ in $a$ (assuming $n \geq 4$) we obtain the quadratic equation
\[ \left((n-4)|C|+2(n+2)\right)na^2-2(n-2)(|C|-n-1)a-(n-6)|C|-6n-4=0, \]
which gives two stationary points (of the right-hand side).
The case $k=2$ is already quite technical but numerical calculations are straightforward.
\end{example}

\subsubsection{A relationship between the covering radius and the minimum distance}

Using again linear programming techniques, we can relate the covering radius and minimum distance of $2k$-designs in the case $D(C) \subseteq \overline{C}$. This will provide better bounds, but the disadvantage is that one needs to have information about the minimum distance of the orthogonal arrays under consideration. 

Let $C \subset F_2^n$ be such that $D(C) \subseteq \overline{C}$, $y \in D(C)$, and let $d=d(C)$ be the minimum distance of $C$.
Let $t_1(y)=\langle y,x_1 \rangle = -1$; i.e., $x_1 \in C$ is such that $d(y,x_1)=n$. Let $x_2 \in C$ be such that $t_2(y)=\langle y,x_2 \rangle$. Since $y$ and $x_1$ differ in all $n$ positions, and $x_1$ and $x_2$ differ in at least $d$ positions, $y$ and $x_2$ coincide in at least these $d$ positions. This means that the distance between $y$ and $x_2$ is at most $n-d$; i.e., 
\[ t_2(y)=\langle y,x_2 \rangle = 1-\frac{2d(y,x_2)}{n} \geq 1-\frac{2(n-d)}{n}=-1+\frac{2d}{n}=t_d. \]

Now, we can consider polynomials 
\begin{equation} \label{rho-d-poly}
f(t)=\left(t+1-\frac{2d}{n}\right)(t-\rho(C))A^2(t),
\end{equation}
where $A(t)$ is a polynomial of degree $k-1$ as above. Utilizing \eqref{def-des-point-wise} for $f$, $C$, and $y$, we obtain
\begin{equation} \label{rho-d-relating-ineq}
f_0|C|-f(-1)= \sum_{i=2}^{|C|} f(t_i(y)) \leq 0,
\end{equation}
where the inequality follows from the fact that $-1+2d/n \leq t_i(y) \leq \rho(C)$ for all $i \in \{2,3,\ldots,|C|\}$. 
This will produce an inequality relating the covering radius $\rho(C)$, minimum distance $d$, and cardinality of $C$.

\begin{example}
In the case $k=1$, a simple calculation with $A(t)=1$ in \eqref{rho-d-poly} (solving $f_0|C|-f(-1) \leq 0$ for $\rho(C)$) gives the bound
\[ \rho(C) \geq \frac{|C|-2d}{(n-2d)|C|+2d}, \]
relating the covering radius $\rho(C)$ and the minimum distance $d$ of any OA$(|C|,n,2)$. 
For $d=1$ this coincides, of course, with \eqref{k=1t=-1}. 
\end{example}

\section{Combining the cases}

In this section, we combine the bounds \eqref{main-bound} and \eqref{q} we obtain a universal upper bound on $\rho(C)$. The universality is in the sense of Levenshtein (cf. \cite[Introduction]{Lev98}) -- there is a branch of the bound for any orthogonal array. 

Clearly, the worse of the two cases will give unconditional bound on $\rho(C)$. 

\begin{theorem} \label{general-theorem}
Let $C \subset F_2^n$ be an OA of strength $2k$. Then 
    \begin{equation} \label{general-bound-rho}
    \rho(C) \geq \min \left\{ t_{k,k}^{0,\ell}, q_k(n,|C|) \right\}.
\end{equation}
In terms of distances, we have  
\begin{equation} \label{general-bound-R}
    R(C) \leq \max \left\{ d_k^{0,\ell}=\frac{n(1-t_{k,k}^{0,\ell})}{2}, \frac{n(1-q_k(n,|C|))}{2} \right\}.
\end{equation}
If the minimum distance of $C$ is at least $d$, then $q_k(n,|C|)$ can be replaced in \eqref{general-bound-rho} and \eqref{general-bound-R} by the quantity coming from \eqref{rho-d-relating-ineq}. 
\end{theorem}

For $|C|>R(n,2k)$, the bound \eqref{general-bound-rho} is always better than the Fazekas-Levenshtein bound. Indeed, both quantities in the minimum in \eqref{general-bound-rho} are larger than $t_k^{0,e}$ from \eqref{FL_bound}. Apart from continuity arguments leading to Theorems \ref{impr_fl} and \ref{case2-thm}, this is confirmed in all our numerical calculations. Of course, in \eqref{general-bound-R} one takes the integer part.

\section{Lower bounds for certain binary OAs}

We construct  three infinite families of binary orthogonal arrays of even strength with  prescribed lower bounds on their covering radii.

Our first family is as follows.
For positive integers $e,m$, let Properties P1 and P2 be defined as follows: 
\begin{itemize}
\item[P1)] For each integer $1 \le j \le e $, there is no positive integer divisor $m_1$ of $m$ such that $m_1 < m$ and
\be%\label{P1}
\left(2^m-1\right) \mid (2j-1)\left(2^{m_1}-1\right).
\nn\ee
\item[P2)] We have
\be%\label{P2}
\sum_{i=0}^{e+1} \binom{2^m-1}{i} > 2^{me}.
\nn\ee
\end{itemize}

Note that if
\be %\label{easyP1}
2e-1 < 2^{\lceil m/2 \rceil},
\nn\ee
then Property P1 above holds.

Next, we present our construction of the first infinite family.

\begin{construction} \label{construction1}
This construction consists of the following steps in order.
\begin{itemize}
\item[i)] Let $e,m$ be positive integers satisfying Properties P1 and P2 above. 
\item[ii)] Let $\theta$ be a primitive root of order $2^m-1$ in $\F_{2^m}$. For $1 \le j \le e$, let $g_j(x) \in \F_2[x]$ be the minimal polynomial of $\theta^{2j-1}$ over $\F_2$.  Let $g(x)=g_1(x) g_2(x) \cdots g_e(x) \in \F_2[x]$. 
\item[iii)] Put $n=2^m-1$ and let $\mathcal{C} \subseteq \F_2^n$ be the cyclic code of length $n$ with generator polynomial $g(x)$. 
\item[iv)] Let $\mathcal{O} \subseteq \F_2^n$ be the Euclidean dual of $\mathcal{C}$. We consider $\mathcal{O}$ as an orthogonal array of length $n$ over $\F_2$.
\end{itemize}
\end{construction}

\begin{remark}
The cyclic code $\mathcal{C}$ in Step iii) of Construction \ref{construction1} is a binary narrow sense BCH code of length $n$ and designed distance $2e+1$ \cite[Chapter 9]{MacWilliamsSloane1977}.
\end{remark}

In the next Theorem, we prove that Construction \ref{construction1} gives an infinite family of binary orthogonal arrays of even strength with a prescribed lower bound on their covering radii.

\begin{Theorem} \label{BCH dual strength even prescribed}
Let $e,m$ be positive integers such that the Properties P1 and P2 above hold. Put $n=2^m-1$. Then Construction \ref{construction1} gives a binary orthogonal array $\mathcal{O} \subseteq \F_{2^n}$ of cardinality $M$ and strength $\tau$ given by
\be
M=2^{me} \;\; \mbox{and} \;\; \tau=2e
\nn\ee
such that its covering radius $R(\mathcal{O})$ satisfies
\be
R(\mathcal{O}) \ge 2^{m-1}-1 -(e-1)2^{m/2}.
\nn\ee
\end{Theorem}
\begin{proof}
We keep the notation of Construction \ref{construction1}. Note that $\mathcal{C}$ is linear over $\F_2$. Using Property P1 we obtain that
$\dim_{\F_2} \mathcal{C}=n-me$. Using \cite[Theorem 2 of Chapter 9 in page 259]{MacWilliamsSloane1977} and Property P2, we obtain that the true minimum distance $d(\mathcal{C})$ of $\mathcal{C}$ is $2e+1$. 

As the minimum distance of $\mathcal{C}$ is $2e+1$, the strength of $\mathcal{O}$ is $2e$, which is a prescribed even integer. Moreover, the cardinality $M$ of $\mathcal{O}$ is
\be
M=2^{me}.
\nn\ee

Next, we obtain a lower bound on the covering radius $R(\mathcal{O})$ of $\mathcal{O}$. First, we present all the elements of $\mathcal{O}$ in a trace representation. For $a_1, a_2, \cdots, a_e \in \F_{2^m}$, let $f(x) \in \F_{2^m}[x]$ be a polynomial of the form
\be \label{trace rep f}
f(x)=a_1 x + a_2 x^3 + \cdots + a_e x^{2e-1} 
\ee
Note the number of polynomials in the form 
(\ref{trace rep f}) 
is equal to $M$. Let $\Tr: \F_{2^m} \ra \F_2$ be the trace map defined as
\be
\Tr(x)=x+x^2+ \cdots + x^{2^{m-1}}.
\nn\ee
For each $f \in \F_{2^m}[x]$ in the form (\ref{trace rep f}), let $\underline{\Tr(f)} \in \F_2^n$ be the codeword defined as
\be  \label{trace rep Trf}
\underline{\Tr(f)} = \left( \Tr(f(\theta)), \Tr(f(\theta^2)), \cdots, \Tr(f(\theta^i)), \cdots, \Tr(f(\theta^{n-1}))\right).
\ee
It is well-known that any codeword of $\mathcal{O}$ is  as in (\ref{trace rep Trf}) for a uniquely determined polynomial $f(x) \in \F_{2^m}[x]$ in the form (\ref{trace rep f}) (see, for example, \cite{GuneriOzbudak2008WeilSerre} or \cite{Wolfmann1989NewBoundsCyclicCodes}).

Let $\bf{1}$ be the vector in $\F_2^n$ given by
\be
{\bf 1}=(1,1, \cdots,1).
\nn\ee
As $\gcd(g(x),x+1)=1$, we have that $g(x) \mid \frac{x^n+1}{x+1}$. This implies that $\bf{1} \in \mathcal{C}$. As $n$ is odd, we have $\bf{1} \cdot \bf{1}=1$, where $\cdot$ is the Euclidean inner product. Hence $\bf{1} \not\in \mathcal{O}$.

 For $f \in \F_{2^m}[x]$ in the form (\ref{trace rep f}), let $d({\bf 1}, \underline{\Tr(f)})$ be the Hamming distance  of  ${\bf 1}$ to the codeword $\underline{\Tr(f)}$
 of the orthogonal array $\mathcal{O}$. 

For $f \in \F_{2^m}[x]$ in the form (\ref{trace rep f}), let $N^*(f)$ denote the cardinality defined as
\be \label{def.N^*}
N^*(f)=| \{\alpha \in \F_{2^m}^*: \Tr(f(\alpha))=0\}|. \nn
\ee

These definitions imply that
\be
d({\bf 1}, \underline{\Tr(f)}) = N^*(f)
\nn\ee
and hence we have
\be \label{lower bound R.1}
R(\mathcal{O}) \ge \min \{N^*(f): \mbox{$f \in \F_{2^m}[x] \setminus \{0\}$ is in the form (\ref{trace rep f})}\}.
\ee

Let $f \in \F_{2^m}[x] \setminus \{0\}$ be in the form 
(\ref{trace rep f}) so that 
$f(x)=a_1 x+ a_2 x^3 + \cdots + a_ex^{2e-1} \in \F_{2^m}$ with 
$(a_1,a_2, \ldots, a_e) \neq (0,0, \ldots, 0)$. Let $\chi(f)$ be the Artin-Schreier curve over $\F_{2^m}$ given by
\be
y^2+y=f(x).
\nn\ee
The genus $g(f)$ of $\chi(f)$ satisfies (\cite[Theorem 3.7.8]{Stichtenoth2009AFFC})
\be
g(f) \le \frac{1}{2}(-2+2e-1+1)=e-1.
\nn\ee

Let $|\chi(f)|$ 
denote the number of 
$\F_{2^m}$-rational points of $\chi(f)$. 
There exists a unique rational point corresponding to the pole of $x$, and there are two rational points corresponding to the zero of $x$. Hence we have (see, for example, \cite{GuneriOzbudak2008WeilSerre} or \cite{Wolfmann1989NewBoundsCyclicCodes})
\be 
|\chi(f)|=1+2+2N^*(f).
\nn\ee
Using Hasse-Weil inequality \cite[Theorem 5.2.3]{Stichtenoth2009AFFC} we obtain
\be
2^m+1-2(e-1)2^{m/2} \le 3 + 2N^*(f) \le 2^m+1 +2(e-1)2^{m/2}.
\nn\ee
The inequality on the left-hand side implies that
\be \label{lower bound R.2}
N^*(f) \ge 2^{m-1}-1-(e-1)2^{m/2}.
\ee

Combining (\ref{lower bound R.1}) and (\ref{lower bound R.2}), we complete the proof.
\end{proof}

We compare the lower bounds from Theorem \ref{BCH dual strength even prescribed} with the linear programming upper bounds (in terms of distances) from Section 4. 

For $e=1$ the codes $\mathcal{O}$ are tight orthogonal arrays of strength 2, i.e. the Fazekas-Levenshtein bound is attained (note that $t_1(y)=-1$ for any $y \in D(C)$ for such arrays). 

When $e=2$, the bounds from Theorem \ref{impr_fl} (conditionally) decrease the gap between the Fazekas-Levenshtein bound and the bounds from Theorem \ref{BCH dual strength even prescribed} for certain small $\ell$. For example, for $n=15$, Theorem \ref{BCH dual strength even prescribed} implies $R(\mathcal{O}) \geq 3$, the Fazekas-Levenshtein bound gives $R(\mathcal{O}) \leq 5$, while  Theorem \ref{impr_fl} with $\ell=-1+4/15=t_2$ leads to the bound $R(\mathcal{O}) \leq 4$. Similarly, in the next case $n=31$, the Fazekas-Levenshtein bound of 12 is (conditionally) improved to 11 (with $\ell=-1+12/31=t_6$), while the lower bound from Theorem \ref{BCH dual strength even prescribed} is 10. Probably, more detailed investigations in these and similar examples can give the exact value of $R(\mathcal{O})$.

Next, we give our second  infinite family of binary orthogonal arrays of even strength with a prescribed lower bound on their covering radii. Here, the strengths are determined as $2$ and $4$ depending on the parity of the integer $m \ge 4$, which is an index of the family.

\begin{construction} \label{construction2}
This construction consists of the following consecutive steps.
\begin{itemize}
\item[i)] Let $m \ge 4$ be an  integer.
\item[ii)] Let $\theta$ be a primitive root of order $2^m-1$ in $\F_{2^m}$. Let $g_1(x) \in \F_2[x]$ be the minimal polynomial of $\theta$ over $\F_2$.
Let $g_{-1}(x) \in \F_2[x]$ be the minimal polynomial of $\theta^{-1}$ over $\F_2$.
Let $g(x)=g_1(x) g_{-1}(x) \in \F_2[x]$. 
\item[iii)] Put $n=2^m-1$ and let $\mathcal{C} \subseteq \F_2^n$ be the cyclic code of length $n$ with generator polynomial $g(x)$. 
\item[iv)] Let $\mathcal{O} \subseteq \F_2^n$ be the Euclidean dual of $\mathcal{C}$. We consider $\mathcal{O}$ as an orthogonal array of length $n$ over $\F_2$.
\end{itemize}
\end{construction}

\begin{remark}
The cyclic code $\mathcal{C}$ in Step iii) of Construction \ref{construction2} is a Melas code of length $n$ over $\F_2$ \cite[Chapter 7]{MacWilliamsSloane1977}.
\end{remark}

In the next Theorem, we prove that Construction \ref{construction2} gives indeed the desired infinite family of binary orthogonal arrays of even strength.

\begin{Theorem} \label{Melas dual strength even prescribed}
Let $m \ge 4$ be an integer. Put $n=2^m-1$. Then Construction \ref{construction2} gives a binary orthogonal array $\mathcal{O} \subseteq \F_{2^n}$ of cardinality $M$ and strength $\tau$ given by
\be
M=2^{2m} \;\; \mbox{and} \;\; \tau=\left\{\begin{array}{rl}
2 & \mbox{if $m$ is even}, \\
4 & \mbox{if $m$ is odd},
\end{array}
\right.
\nn\ee
such that its covering radius $R(\mathcal{O})$ satisfies
\begin{equation} \label{BC2}
R(\mathcal{O}) \ge \left\lceil \frac{2^m-1-\left\lfloor 2^{m/2+1} \right\rfloor}{2}\right\rceil.
\end{equation}
\end{Theorem}
\begin{proof}
We use the methods of Theorem \ref{BCH dual strength even prescribed}.
We keep the notation of Construction \ref{construction2} 
and the proof of Theorem \ref{BCH dual strength even prescribed}. Note that $\mathcal{C}$ is linear over $\F_2$. Let $d$ be the minimum distance of $\mathcal{C}$. Using \cite{MacWilliamsSloane1977} and \cite{SchoofVandervlugt1991},
we obtain that $\dim_{\F_2} \mathcal{C}=n-2m$ and 
\be
d=\left\{\begin{array}{rl}
3 & \mbox{if $m$ is even}, \\
5 & \mbox{if $m$ is odd}.
\end{array}
\right.
\nn\ee
These arguments establish the cardinality and strength of $\mathcal{O}$. It remains to prove the lower bound of $R(\mathcal{O})$.

As $\mathcal{C}$ is a cyclic code, using the methods of the proof of Theorem \ref{BCH dual strength even prescribed}, we obtain that any codeword of $\mathcal{O}$ is given by
\be
\left( \Tr\left(a \theta  + b \theta^{-1}\right), \Tr\left(a \theta^2  + b \theta^{-2}\right), \cdots, \Tr\left(a \theta^i  + b \theta^{-i}\right), \cdots, \Tr\left(a \theta^{n-1}  + b \theta^{-(n-1)}\right)\right),
\nn\ee
where $a,b \in \F_{2^m}$.

For $a,b \in \F_{2^m}$, let $N^*(a,b)$ denote the cardinality defined as
\be 
N^*(a,b)=\left| \left\{\alpha \in \F_{2^m}^*: \Tr\left(a \alpha  + \frac{b}{\alpha}\right)=0\right\}\right|.
\nn\ee

These definitions and the methods of the proof of Theorem \ref{BCH dual strength even prescribed} imply that
\be \label{lower bound R.2a}
R(\mathcal{O}) \ge \min \{N^*(a,b): (a,b) \in \F_{2^m} \times \F_{2^m} \setminus\{(0,0)\}\}.
\ee

If $a \neq 0$ and $b=0$, then we have
\be \label{ep1.Melas dual strength even prescribed}
N^*(a,0)=N^*(1,0)=|\left\{ \alpha \in \F_{2^m}^*: \Tr(\alpha)=0\right\}|=2^{m-1}-1.
\ee
Similarly, if $a=0$ and $b \neq 0$, then we have
\be \label{ep2.Melas dual strength even prescribed}
N^*(0,b)=2^{m-1}-1.
\ee

Assume that $a,b \in \F_{2^m}^*$. Let $\chi(a,b)$ be the Artin-Schreier curve over $\F_{2^m}$ given by
\be
y^2+y=ax + \frac{b}{x}.
\nn\ee
The genus of $\chi(a,b)$ is $1$ (\cite[Theorem 3.7.8]{Stichtenoth2009AFFC}), namely $\chi(a,b)$ is an elliptic curve.

Let $|\chi(a,b)|$ 
denote the number of 
$\F_{2^m}$-rational points of $\chi(a,b)$. 
There exists a unique rational point corresponding to the pole of $x$, and there exists a unique rational point corresponding to the zero of $x$. Hence, as in the proof of Theorem \ref{BCH dual strength even prescribed} we have
\be 
|\chi(a,b)|=2+2N^*(a,b).
\nn\ee
Using Serre's inequality \cite[Theorem 5.3.1]{Stichtenoth2009AFFC} we obtain
\be \label{ep3.Melas dual strength even prescribed}
2 + 2N^*(f) \ge 2^m+1-\lfloor2^{m/2+1} \rfloor.
\ee
We complete the proof combining (\ref{lower bound R.2a}),  (\ref{ep1.Melas dual strength even prescribed}), (\ref{ep2.Melas dual strength even prescribed}), (\ref{ep3.Melas dual strength even prescribed}) and noting that $N^*(a,b)$ is an integer.
\end{proof}

The bounds from Theorem \ref{impr_fl} are close to the lower bound \eqref{BC2} for small $m$ and some $\ell$ which are relatively close to $-1$. For $m=2$ (strength 2) we have $R(\mathcal{O}) \geq 4$ from \eqref{BC2} and the Fazekas-Levenshtein bound says $R(\mathcal{O}) \leq 7$, while Theorem \ref{impr_fl} with $\ell=-1+12/15=t_6$ gives $R(\mathcal{O}) \leq 5$. For $m=3$ (strength 4) we have $R(\mathcal{O}) \geq 10$ from \eqref{BC2} and the Fazekas-Levenshtein bound says $R(\mathcal{O}) \leq 12$, while Theorem \ref{impr_fl} with $\ell=-1+12/31=t_6$ gives $R(\mathcal{O}) \leq 11$.

Finally, we give our third infinite family of binary orthogonal arrays of strength $4$ with a prescribed lower bound  on their covering radii.

\begin{construction} \label{construction3}
This construction consists of the following steps in order.
\begin{itemize}
\item[i)] Let $m \ge 2$ be an integer.
\item[ii)] Let $\theta$ be a primitive root of order $2^{2m}+1$ in $\F_{2^{4m}} \setminus \F_{2^{2m}}$. Let $g(x) \in \F_2[x]$ be the minimal polynomial of $\theta$ over $\F_2$. 
\item[iii)] Put $n=2^{2m}+1$ and let $\mathcal{C} \subseteq \F_2^n$ be the cyclic code of length $n$ with generator polynomial $g(x)$. 
\item[iv)] Let $\mathcal{O} \subseteq \F_2^n$ be the Euclidean dual of $\mathcal{C}$. We consider $\mathcal{O}$ as an orthogonal array of length $n$ over $\F_2$.
\end{itemize}
\end{construction}

\begin{remark}
The cyclic code $\mathcal{C}$ in Step iii) of Construction \ref{construction3} is a Zetterberg code of length $n$ over $\F_2$ \cite{Zetterberg1962}. We note that $\mathcal{C}$ is a quasi-perfect code (see \cite{Dodunekov1985} or \cite{Moreno1983}).
\end{remark}

In the following Theorem below, we prove that Construction \ref{construction3} gives an infinite family of binary orthogonal arrays of strength $4$ with a prescribed lower bound on their covering radii. First we need to prove a technical lemma that we use below.

\begin{lemma} \label{lemma.construction3}
Let $m \ge 2$ be an integer. Put $q=2^{2m}$ and let $H$ be the multiplicative subgroup of $\F_{q^2}^*$ with $|H|=q+1$. For $a_1 \in \F_q^*$, let $S_1(a_1)$ and $S_2(a_1)$ be the multiset subsets in $\F_q$ defined as
\be
S_1(a_1)=\left\{* a_1 \left(x + \frac{1}{x}\right): x \in H *\right\},
\nn\ee
and
\be
S_2(a_1)=\left\{* a_1 \left(y + \frac{1}{y}\right): y \in \F_q^* *\right\}.
\nn\ee
There exists a partition of $\F_q^*$ such that
\be \label{e1.lemma.construction3}
\left\{u_1,u_2, \cdots, u_{q/2}\right\} \bigcup
\left\{v_1,v_2, \cdots ,v_{q/2-1}\right\}=\F_q^*
\ee
satisfying
\be
S_1(a_1)=\left\{* 0^{(1)}, u_1^{(2)}, \ldots, u_{q/2}^{(2)}*\right\},
\nn\ee
and
\be
S_2(a_1)=\left\{* 0^{(1)}, v_1^{(2)}, \ldots, v_{q/2-1}^{(2)}*\right\}.
\nn\ee
Here the term $c^{(i)}$ denotes that $c$ appears exactly $i$ times in the multiset $S_1(a_1)$ (or  $S_2(a_2)$). Moreover, the partition in (\ref{e1.lemma.construction3}) depends on $a_1 \in \F_q^*$.
\end{lemma}
\begin{proof}
Let $\psi_1: H \ra \F_q$ and $\psi_2:\F_q^* \ra \F_q$ be the maps defined as
\be
\psi_1(x)=a_1\left(x + \frac{1}{x}\right) \;\; \mbox{and} \;\;
\psi_2(y)=a_1\left(y + \frac{1}{y}\right)
\nn\ee

If $b \in \Image \psi_1 \setminus \{0\}$, then the cardinality of the preimage
$|\psi^{-1}(b)|$ is exactly $2$. Indeed, there exists $x \in H \setminus 
\{1\}$ such that $\psi_1(x)=b$ as $b \in \Image \psi_1$ and $b \neq 0$. We have
\be
a_1\left(x + \frac{1}{x}\right)=b \Ra x^2 +\frac{b}{a_1}x + 1 =0.
\nn\ee
Therefore $|\psi^{-1}(b)| \le 2$ as a polynomial of degree $2$ can have at most $2$ solution in $\F_{q^2}$. Moreover $\frac{1}{x} \in \psi_1^{-1}(b)$ and $x \neq \frac{1}{x}$, where we use that $x \ne 1$ (which means $b \neq 0$). These arguments complete the proof of the argument that  $|\psi_1^{-1}(b)|=2$ for each $b \in \Image \psi_1 \setminus \{0\}$.

Similarly, we show that $|\psi_2^{-1}(b)|=2$ for each $b \in \Image \psi_2 \setminus \{0\}$.

It remains to show that $\left( \Image \psi_1 \setminus \{0\} \right) \cap \left( \Image \psi_2 \setminus \{0\} \right) = \emptyset$. Assume the contrary. Then there exist $x \in H \setminus \{1\}$  and $y \in \F_q^* \setminus \{1\}$ such that
\be
\begin{array}{rcl}
\dd a_1 \left( x + \frac{1}{x}\right) =a_1 \left(y + \frac{1}{y}\right)  & \Ra & \dd \frac{x^2+1}{x}=\frac{y^2+1}{y} \\ \\
& \Ra & \dd x^2y+y=y^2x+x \\ \\
& \Ra & \dd (x+y)(xy+1)=0 \\ \\
& \Ra & \dd x=\frac{1}{y}.
\end{array}
\nn\ee
Here we use the fact that $\left(H \setminus \{1\}\right) \cap \left(\F_q^* \setminus \{1\}\right) = \emptyset.$ Also if $x \in H \setminus \{1\}$, then $1/x \notin \F_q^*$. This completes the proof.
\end{proof}

Now we are ready for the Theorem.

\begin{Theorem} \label{Zetterberg dual strength even prescribed}
Let $m \ge 2$ be an integer. Put $n=2^{2m}+1$. Then Construction \ref{construction3} gives a binary orthogonal array $\mathcal{O} \subseteq \F_{2^n}$ of cardinality $M$ and strength $\tau$ given by
\be
M=2^{4m} \;\; \mbox{and} \;\; \tau=4
\nn\ee
such that its covering radius $R(\mathcal{O})$ satisfies
\begin{equation} \label{bound-3}
R(\mathcal{O}) \ge 2^{2m-1} - 2^{m} +1.
\end{equation}
\end{Theorem}
\begin{proof}
Using \cite{DodunekovNilsson1992ZetterbergDecoding}, we obtain that
$\mathcal{C}$ is a linear code over $\F_2$ with $\dim_{\F_2} \mathcal{C}=n-4m$ and the minimum distance of $\mathcal{C}$ is $5$. These establish the cardinality of $\mathcal{O}$ and the strength of $\mathcal{O}$. In the rest, we prove the lower bound on $R(\mathcal{O})$.

Put $q=2^{2m}$ and let $H$ be the multiplicative subgroup of $\F_{q^2}^*$ with $|H|=q+1$. Let $\Tr_{q^2/2}: \F_{q^2} \ra \F_2$ and $\Tr_{q/2}: \F_q \ra \F_2$ be the corresponding trace maps from $\F_{q^2}$ and $\F_q$ onto $\F_2$. Note that $\Tr_{q^2/2}(x)=\Tr_{q/2}(x+x^q)$ for each $x \in \F_{q^2}$.

For $a \in \F_{q^2}^*$, let $N^*(a)$ denote the cardinality defined as
\be 
N^*(a)=\left| \left\{h \in H: \Tr_{q^2/2}\left(a h \right)=0\right\}\right|.
\nn\ee
Using the methods of the proof of Theorem \ref{Melas dual strength even prescribed}, we obtain that
\be \label{ep1.Zetterberg dual strength even prescribed}
R(\mathcal{O}) \ge \min \{N^*(a): a \in \F_{q^2}^*\}\}.
\ee

As $\gcd(q-1,q+1)=1$, for a given $a \in \F_{q^2}^*$, there exist unqiely determined elements $a_1 \in \F_q^*$ and $h_1 \in H$ such that $a=a_1h_1$.

For $a_1 \in \F_q^*$ and $h_1,h\in H$, we have
\be
\Tr_{q^2/2}( a_1h_1h) = \Tr_{q/2}\left(a_1\left(h_1h + \frac{1}{h_1h}\right)\right).
\nn\ee

For $a_1 \in \F_{q}^*$, let $N_1^*(a_1)$ denote the cardinality defined as
\be 
N_1^*(a_1)=\left| \left\{x \in H: \Tr_{q/2}\left(a_1 \left(x + \frac{1}{x}\right) \right)=0\right\}\right|.
\nn\ee

Using these arguments, we obtain that (\ref{ep1.Zetterberg dual strength even prescribed}) is equivalent to
\be \label{ep2.Zetterberg dual strength even prescribed}
R(\mathcal{O}) \ge \min \{N_1^*(a_1): a_1 \in \F_{q}^*\}\}.
\ee

For $a_1 \in \F_q^*$, let $S_1(a_1)$ and $S_2(a_1)$ be the multiset subsets in $\F_q$ defined as
\be
S_1(a_1)=\left\{* a_1 \left(x + \frac{1}{x}\right): x \in H *\right\},
\nn\ee
and
\be
S_2(a_1)=\left\{* a_1 \left(y + \frac{1}{y}\right): y \in \F_q^* *\right\}.
\nn\ee

Using Lemma \ref{lemma.construction3} we obtain a partition of $\F_q^*$ such that
\be \label{ep3.Zetterberg dual strength even prescribed}
\left\{u_1,u_2, \cdots, u_{q/2}\right\} \bigcup
\left\{v_1,v_2, \cdots ,v_{q/2-1}\right\}=\F_q^*
\ee
and
\be 
N_1^*(a_1)= 1 + 2 \left| \left\{1 \le i \le q/2: \Tr_{q/2}(u_i)=0\right\}\right|.
\nn\ee

Using the partition in (\ref{ep3.Zetterberg dual strength even prescribed}),
let $N_2^*(a_1)$ be the integer defined as
\be 
N_2^*(a_1)= 1 + 2 \left| \left\{1 \le i \le q/2-1: \Tr_{q/2}(v_i)=0\right\}\right|.
\nn\ee
These definitions and Lemma \ref{lemma.construction3} imply that we have
\be \label{ep4.Zetterberg dual strength even prescribed}
N_1^*(a_1)+N_2^*(a_1)=q
\ee
for each $a_1 \in \F_q^*$. Combining (\ref{ep2.Zetterberg dual strength even prescribed}) and  (\ref{ep4.Zetterberg dual strength even prescribed})
we obtain
\be \label{ep5.Zetterberg dual strength even prescribed}
R(\mathcal{O}) \ge q-\max \{N_2^*(a_1): a_1 \in \F_{q}^*\}\}.
\ee

For $a_1 \in \F_q^*$, let $\chi^{(2)}(a_1)$ be the Artin-Schreier curve over $\F_{2^{2m}}$ given by
\be
y^2+y=a_1x + \frac{a_1}{x}.
\nn\ee
The genus of $\chi^{(2)}(a)$ is $1$ (\cite[Theorem 3.7.8]{Stichtenoth2009AFFC}), and hence $\chi^{(2)}(a_1)$ is an elliptic curve as in the proof of Theorem \ref{Melas dual strength even prescribed}.

For $a_1 \in \F_q^*$, let $|\chi^{(2)}(a_1)|$ 
denote the number of 
$\F_{2^{2m}}$-rational points of $\chi^{(2)}(a_1)$. Using the methods of Theorem \ref{Melas dual strength even prescribed},  we have
\be 
|\chi^{(2)}(a_1)|=2+2N_2^*(a_1).
\nn\ee
Using Hasse-Weil inequality \cite[Theorem 5.2.3]{Stichtenoth2009AFFC} we obtain
\be \label{ep6.Zetterberg dual strength even prescribed}
2 + 2N_2^*(a_1) \le 2^{2m}+1+2^{m+1}.
\ee
We complete the proof combining (\ref{ep5.Zetterberg dual strength even prescribed}) and   (\ref{ep6.Zetterberg dual strength even prescribed}).
\end{proof}

Theorem \ref{impr_fl} improve the Fazekas-Levenshtein bounds by 1 (in terms of distances) for $\ell=-1+2/n=t_1$ and so does Theorem \ref{general-theorem}. Improvements for reasonable larger $\ell$ do not reflect the integer part. Thus, in the first case $m=2$ (strength 4) we have $R(\mathcal{O}) \geq 5$ from \eqref{bound-3} and the Fazekas-Levenshtein bound says $R(\mathcal{O}) \leq 6$, while Theorem \ref{general-theorem} with $\ell=-1+2/17=t_1$ gives $R(\mathcal{O}) \leq 5$ (the exact value of the covering radius is found). For $m=3$ (strength 4 again) we have $R(\mathcal{O}) \geq 25$ from \eqref{bound-3} while the Fazekas-Levenshtein bound gives $R(\mathcal{O}) \leq 28$ and Theorem \ref{general-theorem} with $\ell=-1+2/65=t_1$ gives $R(\mathcal{O}) \leq 27$.

\bigskip

{\bf Funding.} The research of the first author was supported by project 
IC-TR/10/2024-2025. The research of the second author is supported by T\"{U}B\.{I}TAK under Grant 223N065.
The research of the third author is supported by the Bulgarian NSF grant KP-06-N72/6-2023.


\begin{thebibliography}{99}

%\bibitem{BBD} S. Boumova,  P. Boyvalenkov, D. Danev, Necessary conditions
%for existence of some designs in polynomial metric spaces, {\it Europ. J. Combin.} 20, 1999, 213-225.
\bibitem{AHLL99} 
A. Ashikhmin, I. Honkala, T. Laihonen, S. Litsyn, 
On relations between covering radius and dual distance, 
{\it IEEE Trans. Inform. Theory}, 45, 1999, 1808--1816.

\bibitem{Ban21} Ei. Bannai, Et. Bannai, T. Ito, R. Tanaka, {\it Algebraic Combinatorics}, Berlin, Boston, De Gruyter, 2021. 

\bibitem{BDHSS-DCC2019} P. Boyvalenkov, P. Dragnev, D. Hardin, E. Saff, M. Stoyanova, On spherical codes with inner products in prescribed interval, {\it Designs Codes Cryptography} 87,
299–315 (2019). 

\bibitem{BDHSS-ieee} P. Boyvalenkov, P. Dragnev, D. Hardin, E. Saff, M. Stoyanova, Universal bounds for size and energy of codes of 
given minimum and maximum distances, {\it IEEE Transactions on Information Theory}, "From Deletion-Correction to Graph Reconstruction: In Memory of Vladimir I. Levenshtein", 67(6), 2021, 3569-3584 (arXiv:1910.07274).

\bibitem{BS21} P. Boyvalenkov, M. Stoyanova, Linear programming bounds for covering radius of spherical designs, \emph{Results in Mathematics}, 76, art. 95 (19 pages), 2021. 

\bibitem{BMS} P. Boyvalenkov, T. Marinova, M. Stoyanova, Nonexistence of few binary orthogonal arrays,
{\it Discrete Applied Mathematics}, 217(P2), 2017, 144--150.

\bibitem{CW79} J. Carter, M. Wegman, Universal classes of hash functions. \emph{Journal of Computer and System Sciences} 18, 1979, 143--154.

\bibitem{DM93} E. Dawson, E. Mahmoodian, Orthogonal arrays and ordered threshold schemes, \emph{Australasian Journal of Combinatorics}, 8, 1993, 27--44.

\bibitem{Del-dis}
{P.\,Delsarte}, {\it An Algebraic Approach to the Association Schemes in Coding Theory},
Philips Res. Rep. Suppl. 10, 1973.

\bibitem{Del73}
{P.\,Delsarte}, Four fundamental parameters of a code and their combinatorial significance,
{\it Inform. Contr.} 23 (1973) 407-438.

%%%%%%%Ferruh 3

\bibitem{Dodunekov1985}
S.~M.~Dodunekov,
\newblock ``The optimal double-error correcting codes of Zetterberg and Dumer--Zinoviev are quasiperfect,''
\newblock \emph{C.R. de l'Acad\'emie bulgare des Sciences}, tome~38, nr.~9,
pp.~1121--1123, 1985.

\bibitem{DodunekovNilsson1992ZetterbergDecoding}
S.~M.~Dodunekov and J.~E.~M.~Nilsson,
\newblock ``Algebraic Decoding of the Zetterberg Codes,''
\newblock \emph{IEEE Transactions on Information Theory},
vol.~38, no.~5, pp.~1570--1573, Sept.~1992.
%%%%%%end of Ferruh 3

\bibitem{FL} G. Fazekas, V. I. Levenshtein, On upper bounds for
code distance and covering radius of designs in polynomial metric
spaces, {\it J. Comb. Theory A}, 70, 267-288 (1995).

\bibitem{GS26} A. L. Gavrilyuk, S. Suda, Extremal orthogonal arrays, arXiv:2512.23459. 


\bibitem{GSV} A. L. Gavrilyuk, S. Suda, J. Vidali, On tight 4-designs in Hamming association schemes, {\it Combinatorica}, 40 (2020) 345--362.

\bibitem{GS08} K. Gopalakrishnan, D. Stinson, Applications of orthogonal arrays to computer science. Ramanujan Mathematical Society, Lecture Notes Series in Mathematics, 7, 2008, 149--164.
%%%%ferruh

\bibitem{GuneriOzbudak2008WeilSerre}
C.~G\"uneri, F.~\"Ozbudak,
``Weil--Serre Type Bounds for Cyclic Codes,''
\emph{IEEE Transactions on Information Theory},
vol.~54, no.~12, pp.~5381--5395, Dec.~2008.
%%%%ferruh

\bibitem{HSS} A. Hedayat, N. J. A. Sloane, J. Stufken, Orthogonal Arrays: Theory and Applications. Springer-Verlag, New York, 1999.

\bibitem{Lev92} V. I. Levenshtein, Designs as maximum codes in polynomial
metric spaces, {\it Acta Appl. Math.} 25, 1992, 1--82.

\bibitem{Lev95}
V. I. Levenshtein, Krawtchouk polynomials and universal bounds for codes and designs in Hamming spaces, {\it IEEE Trans. Inform. Theory} 41, 1995, 1303--1321.

\bibitem{Lev98} V. I. Levenshtein, Universal bounds for codes and designs,
Chapter 6 (499-648) in {\it Handbook of Coding Theory}, Eds. V.Pless and
W.C.Huffman, Elsevier Science B.V., 1998.

\bibitem{LSS98} S. Litsyn, P. Sol\'e, R.  Struik, On the covering radius of an unrestricted code as a function of the rate and dual distance, {\it Discr. Appl. Math.}, 82, 1998, 177--191.

\bibitem{LT96} S. Litsyn,A. Tiet{\"a}v{\"a}inen, Upper bounds on the covering radius of a code with a given dual distance, 
{\it Europ. J. Combin.}, 17, 1996, 265--270.

%%% addition  Ferruh
\bibitem{MacWilliamsSloane1977}
F.~J.~MacWilliams and N.~J.~A.~Sloane,
\emph{The Theory of Error-Correcting Codes},
North-Holland Mathematical Library, Vol.~16,
North-Holland Publishing Company, Amsterdam--New York--Oxford, 1977.

\bibitem{Moreno1983}
O.~Moreno,
\newblock ``Further results on quasiperfect codes related to Goppa codes,''
\newblock \emph{Congressus Numerantium}, vol.~40, pp.~249--256, 1983.

%%% end addition  Ferruh

\bibitem{MK94} R. Mukerjee, S. Kageyama, On existence of two symbol complete orthogonal arrays, {\it J. Combin. Theory Ser. A}, 66, 1994, 176--181.

\bibitem{N79} R. Noda, On orthogonal arrays of strength 4 achieving Rao's bound, {\it J. London Math. Soc.} (2), 19, 1979, 385--390.

%
%
%\bibitem{NN} S. Nikova, V. Nikov, Improvement of the Delsarte bound for t-designs when it is not the best bound possible,
%{\it Des., Codes Crypt.} 28, 2003, 201-222.

\bibitem{Rao} C. R. Rao, Factorial experiments derivable from combinatorial arrangements of arrays,
\emph{J. Royal Stat. Soc.} 89 (1947) 128--139.

%%%%%%%%%%Ferruh2

\bibitem{SchoofVandervlugt1991}
R.~Schoof and M.~van~der~Vlugt,
\newblock ``Hecke Operators and the Weight Distributions of Certain Codes,''
\newblock \emph{Journal of Combinatorial Theory, Series A}, vol.~57,
pp.~163--186, 1991.

%%%%%%%%%%Ferruh2

\bibitem{sol95} P. Sol\'e, Packing radius, covering radius, and dual distance, {\it IEEE Trans. Inform. Theory}, 41, 1995, 268--272.

\bibitem{SS93} P. Sol\'e, P. Stokes, Covering radius, codimension, and dual-distance width, {\it IEEE Trans. Inform. Theory}, 39, 1993, 1195--1203.

%%%%%%%%%%Ferruh

\bibitem{Stichtenoth2009AFFC}
H.~Stichtenoth,
\emph{Algebraic Function Fields and Codes},
Graduate Texts in Mathematics, Vol.~254,
Springer-Verlag, Berlin--Heidelberg, 2nd ed., 2009.
%%%%%%%%%%%%%%Ferruh


\bibitem{Sti94} D. Stinson, Combinatorial techniques for universal hashing, \emph{Journal of Computer and System Sciences}, 48, 1994, 337--346.

\bibitem{Tie90} A. Tiet{\"a}v{\"a}inen, An upper bound on the covering radius as a function of the dual distance, {\em IEEE Trans. Inform. Theory}, 36 (1990) 1472--1474.
	
\bibitem{Tie91} A. Tiet{\"a}v{\"a}inen, Covering radius and dual distance, {\em Des. Codes Cryptogr.} {\bf 1} (1991) 31--46.

%%%%%%ferruh

\bibitem{Wolfmann1989NewBoundsCyclicCodes}
J.~Wolfmann,
``New bounds on cyclic codes from algebraic curves,''
in \emph{Lecture Notes in Computer Science}.
New York: Springer-Verlag, 1989, vol.~388, pp.~47--62.

%%%%%%%%%ferruh

\bibitem{Sze} G. Szeg\"{o}, {\it Orthogonal polynomials},
{AMS Col. Publ.}, vol. 23, Providence, RI, 1939.

%%%%%%%%%%ferruh3

\bibitem{Zetterberg1962}
L.~H.~Zetterberg,
\newblock ``Cyclic codes from irreducible polynomials for correction of multiple errors,''
\newblock \emph{IRE Trans. Inform. Theory}, vol.~IT-8, pp.~13--20, 1962.

%%%%%%%%%%%%ferruh3

\end{thebibliography}
\end{document}